\documentclass[a4paper,11pt]{article}
\usepackage{microtype} 
\usepackage{authblk}
\RequirePackage{footmisc}

\usepackage{graphicx}
\usepackage{xspace}
\usepackage{amsmath,amssymb}
\usepackage[caption=false]{subfig}
\usepackage{xcolor}
\usepackage{hyperref}
\usepackage{amsthm,amsfonts,amssymb,amsmath,array}
\usepackage{fullpage}

\usepackage{thmtools,thm-restate}
\usepackage{cleveref}

\newtheorem{theorem}{Theorem}
\newtheorem{lemma}[theorem]{Lemma}
\newtheorem{definition}[theorem]{Definition}
\newtheorem{corollary}[theorem]{Corollary}

\graphicspath{{./figures/}}

\title{Generalized Sweeping Line Spanners}

\author[1]{Keenan Lee\thanks{Email: klee6239@uni.sydney.edu.au}}
\author[1]{Andr\'e van Renssen\thanks{Email: andre.vanrenssen@sydney.edu.au}}

\affil[1]{University of Sydney, Australia}

\date{}

\newcommand{\etal}{\emph{et~al.}\xspace}

\begin{document}

\maketitle

\begin{abstract}
We present \emph{sweeping line graphs}, a generalization of $\Theta$-graphs. We show that these graphs are spanners of the complete graph, as well as of the visibility graph when line segment constraints or polygonal obstacles are considered. Our proofs use general inductive arguments to make the step to the constrained setting. These same arguments could apply to other spanner constructions in the unconstrained setting, removing the need to find separate proofs that they are spanning in the constrained and polygonal obstacle settings.
\end{abstract}

\section{Introduction}
A \emph{geometric graph} $G$ is a graph whose vertices are points in the plane and whose edges are line segments between pairs of points. Every edge in the graph has weight equal to the Euclidean distance between its two endpoints. The distance between two vertices $u$ and $v$ in $G$, denoted by $\delta_{G}(u,v)$, or simply $\delta(u,v)$ when $G$ is clear from the context, is defined as the sum of the weights of the edges along a shortest path between $u$ and $v$ in $G$. A subgraph $H$ of $G$ is a $t$-spanner of $G$ (for $t \ge 1$) if for each pair of vertices $u$ and $v$, $\delta_{H}(u,v) \le t \cdot \delta_{G}(u,v)$. The smallest $t$ for which $H$ is a $t$-spanner is the \emph{spanning ratio} or \emph{stretch factor} of $H$. The graph $G$ is referred to as the \emph{underlying graph} of $H$ and, unless otherwise specified, this is assumed to be the complete Euclidean graph. The spanning properties of various geometric graphs have been studied extensively in the literature (see~\cite{BS11,NS-GSN-06} for an overview). The spanning ratio of a class of graphs is the supremum over the spanning ratios of all members of that graph class. Since spanners preserve the lengths of all paths up to a factor of $t$, these graphs have applications in the context of geometric problems, including motion planning and optimizing network costs and delays. 

We introduce a generalization of an existing geometric spanner ($\Theta$-graphs) which we call \emph{sweeping line graphs}. We show that these graphs are spanners, also when considering line segment obstacles or polygonal obstacles during their construction. We show this using a very general method that we conjecture applies to other geometric spanners as well, meaning that such proofs do not have to be reinvented for different spanner constructions. 

Clarkson~\cite{C87} and Keil~\cite{K88} independently introduced $\Theta$-graphs as follows: for each vertex $u$, we partition the plane into $k$ disjoint cones with apex $u$. Each cone has aperture equal to $\theta = \frac{2\pi}{k}$ (see Figure~\ref{fig:cones}) and the orientation of these cones is identical for every vertex. The $\Theta$-graph is constructed by, for each cone with apex $u$, connecting $u$ to the vertex $v$ whose projection onto the bisector of the cone is closest to $u$ (see Figure~\ref{fig:closest}). When $k$ cones are used, we denote the resulting $\Theta$-graph as the $\Theta_{k}$-graph. We note that while $\Theta$-graphs can be defined using an arbitrary line in the cone instead of the bisector (see Chapter 4 of \cite{NS-GSN-06}), the bisector is by far the most commonly used construction and the spanning ratios mentioned in the following paragraph only apply when the bisector is used in the construction process. 

\begin{figure}[h]
  \begin{minipage}[b]{0.45\linewidth}
    \centering
    \includegraphics{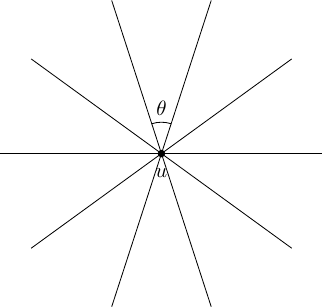}
    \caption{The plane around $u$ is split into 10 cones.}
    \label{fig:cones}
  \end{minipage}
  \hspace{0.5cm}
  \begin{minipage}[b]{0.45\linewidth}
    \centering
    \includegraphics{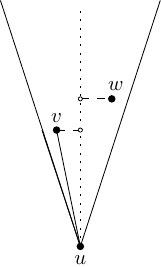}
    \caption{Vertex $v$ is the vertex with the projection closest to $u$.}
    \label{fig:closest}
  \end{minipage}
\end{figure}

Ruppert and Seidel~\cite{RS91} upperbounded the spanning ratio of these graphs (when there are at least 7 cones) by $1/(1 - 2 \sin (\theta/2))$, when $\theta < \pi/3$. Bonichon~\etal~\cite{BGHI10} showed that the $\Theta_6$-graph has a tight spanning ratio of 2, i.e., it has a matching upper and lower bound of 2. Other recent results include a tight spanning ratio of $1 + 2 \sin(\theta/2)$ for $\Theta$-graphs with $4m + 2$ cones, where $m \geq 1$ and integer, and improved upper bounds for the other families of $\Theta$-graphs~\cite{BCMRV16}. When there are fewer than 6 cones, most inductive arguments break down. Hence, it was only during the last decade that upper bounds on the spanning ratio of the $\Theta_5$-graph and the $\Theta_4$-graph were determined: $\sqrt{50 + 22 \sqrt{5}} \approx 9.960$ for the $\Theta_5$-graph~\cite{BMRV2015} and $(1 + \sqrt{2}) \cdot (\sqrt{2} + 36) \cdot \sqrt{4 + 2 \sqrt{2}} \approx 237$ for the $\Theta_4$-graph~\cite{BBCRV2013}. These bounds were recently improved to $5.70$ for the $\Theta_5$-graph~\cite{BHO21} and $17$ for the $\Theta_4$-graph~\cite{BCDS19}. Constructions similar to those demonstrated by El Molla~\cite{E09} for Yao-graphs show that $\Theta$-graphs with fewer than 4 cones are not spanners. In fact, until recently it was not known that the $\Theta_3$-graph is connected~\cite{ABBBKRTV2013}. 

An alternative way of describing the $\Theta$-graph construction is that for each cone with apex $u$, we visualize a line perpendicular to the bisector of the cone that sweeps outwards from the apex $u$. The $\Theta$-graph is then constructed by connecting $u$ to the first vertex $v$ in the cone encountered by the sweeping line as it moves outwards from $u$ (see Figure~\ref{fig:sweepingTheta}).

\begin{figure}[h]
  \begin{minipage}[t]{0.45\linewidth}
    \centering
    \includegraphics{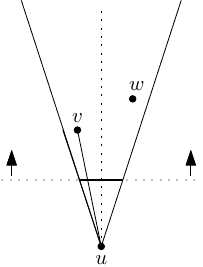}
    \caption{The sweeping line of a cone in a $\Theta$-graph. The sweeping line is a thick black segment inside the cone and grey dotted outside, as vertices outside the cone as ignored.}
    \label{fig:sweepingTheta}
  \end{minipage}
  \hspace{0.5cm}
  \begin{minipage}[t]{0.45\linewidth}
    \centering
    \includegraphics{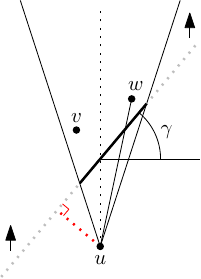}
    \caption{The sweeping line of a cone in the sweeping line graph. For comparison to $\Theta$-graphs, the line through $u$ perpendicular to the sweeping line is shown in red.}
    \label{fig:sweepingLine}
  \end{minipage}
\end{figure}

The \emph{sweeping line graph} generalizes this construction by allowing the sweeping line to be at an angle $\gamma$ to the line perpendicular to the bisector of the cone (see Figure~\ref{fig:sweepingLine}, a more precise definition follows in Section 2). When $\gamma \in [0, \frac{\pi - 3\theta}{2})$ we show that the resulting graph is a spanner whose spanning ratio depends only on $\theta$ and $\gamma$. We note that this angle $\gamma$ implies that the line perpendicular to the sweeping line can be \emph{outside} the cone associated with that sweeping line, which is not supported in the common $\Theta$-graph (using bisectors) or the more general ones described in~\cite{NS-GSN-06}. For example, when 10 cones are used (i.e., $\theta = \pi/5$) the construction in~\cite{NS-GSN-06} allows for an angle $\gamma$ up to $\theta/2 = \pi/10$, while our construction allows the far larger value of $7\pi/20$ and this difference increases as the number of cones increases (i.e., $\theta$ decreases). The ability to rotate the sweeping line has previously been used to define ordered $\Theta$-graphs~\cite{BGM04}, where a sweeping line perpendicular to a cone boundary gave a more efficient construction method. 

Pushing the generalization of these graphs even further, we consider spanners in two more general settings by introducing line segment constraints and polygonal obstacles. Specifically, given a set $P$ of points in the plane, let $S$ be a set of line segments connecting pairs of points in $P$ (not every vertex in $P$ needs to be an endpoint of a line segment in $S$). We refer to $S$ as \emph{line segment constraints}, or simply \emph{constraints}. The set of line segment constraints is planar, i.e. no two constraints intersect properly, but they can share endpoints. Two vertices $u$ and $v$ can see each other if and only if either the line segment $uv$ does not properly intersect any constraint or $uv$ is itself a constraint. If two vertices $u$ and $v$ can see each other, the line segment $uv$ is a \emph{visibility edge} (this notion can be generalized to apply to arbitrary points that can see each other). The \emph{visibility graph} of $P$ with respect to a set of constraints $S$, denoted by Vis$(P,S)$, has $P$ as vertex set and all visibility edges defined by vertices in $P$ as edge set. In other words, it is the complete graph on $P$ minus all edges that properly intersect one or more constraints in $S$. The aim of this generalization is to construct a spanner such that no edge of the spanner properly intersects any constraint. In other words, to construct a spanner of Vis$(P,S)$. 

Polygonal obstacles generalize the notion of line segment constraints by allowing the constraints to be simple polygons instead of line segments, i.e., by excluding the internal area from being accessible. In this situation, $S$ is a finite set of simple polygonal obstacles where each corner of each obstacle is a point in $P$, such that no two obstacles intersect. We assume that each vertex is part of at most one polygonal obstacle and occurs at most once along its boundary, i.e., the obstacles are vertex-disjoint simple polygons, and no polygon contains any point of $P$ in its interior. Note that $P$ can also contain vertices that do not lie on the corners of the obstacles. The definitions of visibility edge and visibility graph are analogous to the ones for line segment constraints.

In the context of motion planning amid obstacles, Clarkson~\cite{C87} showed how to construct a linear-sized $(1+\epsilon)$-spanner of Vis$(P,S)$. Subsequently, Das~\cite{D97} showed how to construct a spanner of Vis$(P,S)$ with constant spanning ratio and constant degree. More recently, the constrained $\Theta_6$-graph was shown to be a 2-spanner of Vis$(P,S)$~\cite{BFRV12} when considering line segment constraints. This result was recently generalized to polygonal obstacles~\cite{polygonallemma}. Most related to this paper is the result by Bose and van Renssen~\cite{BR2019}, who generalized the results from Bose~\etal~\cite{BCMRV16} to the setting with line segment constraints, without increasing the spanning ratios of the graphs.

In this paper, we examine the sweeping line graph in the setting without constraints (or the unconstrained setting), with line segment constraints (or the constrained setting), and with polygonal obstacles. First, we generalize the spanning proof of the $\Theta$-graph given in the book by Narasimhan and Smid~\cite{NS-GSN-06} to the sweeping line graph in the unconstrained setting. Next, we extend the proof to the constrained setting and finally apply it to the case of polygonal obstacles. In all three cases, we prove that the spanning ratio is upperbounded by $\frac{1}{\cos(\frac{\theta}{2} + \gamma) - \sin\theta}$, where $\theta = \frac{2\pi}{k}$ ($k \geq 7$) and $\gamma \in [0, \frac{\pi - 3\theta}{2})$. The most interesting aspect of our approach is that the step from the unconstrained to the constrained and polygonal obstacle settings is very general and could apply to other spanner constructions in the unconstrained setting as well, making it a step towards a general condition of which spanners in the unconstrained setting can be proven to be spanners in the presence of obstacles.

\section{Preliminaries}
Throughout this paper, the notation $|pq|$ refers to the Euclidean distance between $p$ and $q$. We also emphasize that a point can be any point in the plane, while a vertex is restricted to being one of the points in the point set $P$. 

Before we formally define the \emph{sweeping line graph}, we first need a few other notions. A cone is the region in the plane between two rays that emanate from the same point, called the apex of the cone. Let $k \ge 7$ and define $\theta = \frac{2\pi}{k}$. If we rotate the positive $x$-axis by angles $i \cdot \theta$, $0 \le i < k$, then we get $k$ rays. Each pair of consecutive rays defines a cone whose apex is at the origin. We denote the collection of these $k$ cones by $\zeta_{k}$. Let $C$ be a cone of $\zeta_{k}$. For any point $p$ in the plane, we define $C_{p}$ to be the cone obtained by translating $C$ such that its apex is at $p$. 

Next, given an angle $\theta$ for a particular cone and an angle $\gamma \in [0, \frac{\pi - 3\theta}{2})$, we give the definition of the sweeping line: For any vertex $x$ in a cone, let the \emph{sweeping line} be the line through the vertex $x$ that is at an angle of $\gamma$ to the line perpendicular to the bisector of the cone. We then define $x_{\gamma}$ to be the intersection of the left-side of the cone and this sweeping line (see Figure~\ref{fig:definition}).

\begin{figure}[h]
\includegraphics{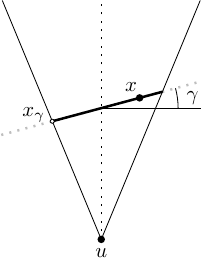}
\centering
\caption{Defining the point $x_{\gamma}$ for any vertex $x$.}
\label{fig:definition}
\end{figure}

Finally, we define the sweeping line graph:
\begin{definition}[Sweeping line graph] Given a set of points $P$ in the plane, an integer $k \ge 7$, $\theta = \frac{2\pi}{k}$, and $\gamma \in [0, \frac{\pi - 3\theta}{2})$. The sweeping line graph is defined as follows:
\begin{enumerate}
    \item The vertices of the graph are the points of $P$.
    \item For each vertex $u$ of $P$ and for each cone $C$ of $\zeta_{k}$, such that the translated cone $C_{u}$ contains one or more vertices of $P \setminus \{u\}$, the spanner contains the undirected edge $(u, r)$, where $r$ is the vertex in $C_{u} \cap P \setminus \{u\}$, which minimizes $|ur_{\gamma}|$. This vertex $r$ is referred to as the \emph{closest} vertex in this cone of $u$. 
\end{enumerate}
\end{definition}
In the remainder of the paper, we use $|ur_{\gamma}|$ when measuring the closeness between a vertex $r$ and the apex $u$ of a cone that contains it.

For ease of exposition, we only consider point sets in general position: no two vertices lie on a line parallel to one of the rays that define the cones, no two vertices lie on a line parallel to the sweeping line of any cone, and no three points are collinear. These assumptions can be removed using standard perturbation techniques. 

We note that the construction time for the sweeping line graph is identical to that of the $\Theta$-graph: $O(n \log n)$ time. This can for example be achieved using the construction algorithm by Keil~\cite{K88} and using the slanted sweepline instead of the perpendicular one (or alternatively rotating the point set). 

Using the structure of the sweeping line graph, we define a simple algorithm called \emph{sweeping-routing} to construct a path between any two vertices $s$ and $t$. The name comes from the fact that this is also a 1-local routing algorithm on the sweeping line graph. Let $t$ be the destination of the routing algorithm and let $u$ be the current vertex (initially $u = s$). If there exists a direct edge to $t$, follow this edge. Otherwise, follow the edge to the closest vertex in the cone of $u$ that contains $t$.

\subsection{Auxiliary Lemmas}
In order to prove that the sweeping line graph is indeed a spanner, we start with a number of auxiliary lemmas needed to prove the main geometric lemma used throughout our proofs.

\begin{lemma}
\label{lem:lemma1}
Let $\theta \in (0, \frac{2\pi}{7}]$ and $\gamma \in [0, \frac{\pi - 3\theta}{2})$. Then $\cos(\frac{\theta}{2} + \gamma) - \sin \theta > 0$.
\end{lemma}
\begin{proof}
Since $\cos(\frac{\theta}{2} + \gamma) - \sin\theta$ is decreasing with respect to $\gamma$ in our domain, it is minimized when $\gamma$ is maximized. It follows that:
\begin{align*}
\cos\left(\frac{\theta}{2} + \gamma\right) - \sin \theta &> \cos\left(\frac{\theta}{2} + \frac{\pi}{2} - \frac{3\theta}{2}\right) - \sin \theta \\
&= \cos\left(\frac{\pi}{2} - \theta\right) - \sin\theta\\
&= \sin\theta - \sin\theta\\
&= 0
\end{align*}
Therefore, within our domain, $\cos(\frac{\theta}{2} + \gamma) - \sin \theta > 0$.
\end{proof}

\begin{lemma}
\label{lem:lemma2}
Let $\theta \in (0, \frac{2\pi}{7}]$, $\gamma \in [0,\frac{\pi - 3\theta}{2})$, and $\kappa \in [0,\theta]$. Then
$\cos(\frac{\theta}{2} - \gamma - \kappa) > 0$ and  $\cos(\frac{\theta}{2} + \gamma - \kappa) > 0$.
\end{lemma}
\begin{proof}
To prove this, it suffices to show that $-\frac{\pi}{2} < \frac{\theta}{2} - \gamma - \kappa, \frac{\theta}{2} + \gamma - \kappa < \frac{\pi}{2}$ as within this domain, $\cos(\frac{\theta}{2} - \gamma - \kappa)$ and  $\cos(\frac{\theta}{2} + \gamma - \kappa)$ are greater than 0.

\emph{Proof of $-\frac{\pi}{2} < \frac{\theta}{2} - \gamma - \kappa < \frac{\pi}{2}$:}

First, we show that $\frac{\theta}{2} - \gamma - \kappa$ is upperbounded by $\frac{\pi}{2}$:
\begin{align*}
    \frac{\theta}{2} - \gamma - \kappa \le \frac{\theta}{2} &\le \frac{\pi}{7} \text{ (using the domain of $\theta$)}\\
    &< \frac{\pi}{2}
\end{align*}

Next, we show that $\frac{\theta}{2} - \gamma - \kappa$ is lowerbounded by $-\frac{\pi}{2}$:
\begin{align*}
    \frac{\theta}{2} - \gamma - \kappa > -\gamma - \kappa &\ge -\frac{\pi - 3\theta}{2}- \theta \text{ (using the domain of $\gamma$ and $\kappa$)}\\
    &> -\frac{\pi}{2}
\end{align*}

\emph{Proof of $-\frac{\pi}{2} < \frac{\theta}{2} + \gamma - \kappa < \frac{\pi}{2}$:}

First, we show that $\frac{\theta}{2} + \gamma - \kappa$ is upperbounded by $\frac{\pi}{2}$:
\begin{align*}
    \frac{\theta}{2} + \gamma - \kappa \le \frac{\theta}{2} + \gamma &< \frac{\theta}{2} + \frac{\pi - 3\theta}{2} \text{ (using the bounds on $\gamma$)}\\
    &= \frac{\pi -2\theta}{2}\\
    &< \frac{\pi}{2}
\end{align*}

Next, we show that $\frac{\theta}{2} + \gamma - \kappa$ is lowerbounded by $-\frac{\pi}{2}$:
\begin{align*}
    \frac{\theta}{2} + \gamma - \kappa > -\kappa &\ge -\frac{2\pi}{7} \text{ (using the bounds on $\kappa$)}\\
    &> -\frac{\pi}{2}
\end{align*}
\end{proof}

\begin{lemma}
\label{lem:lemma3}
Let $a$ and $b$ be positive reals and $\theta \in (0, \frac{2\pi}{7}]$, $\gamma \in [0,\frac{\pi - 3\theta}{2})$, and $\kappa \in [0,\theta]$. Then
$a-\frac{b(\cos(\frac{\theta}{2} + \gamma) - \sin \theta)}{\cos (\frac{\theta}{2} - \gamma - \kappa)} \le a - b(\cos(\frac{\theta}{2} + \gamma) - \sin \theta)$ and $a-\frac{b(\cos(\frac{\theta}{2} + \gamma) - \sin \theta)}{\cos (\frac{\theta}{2} + \gamma - \kappa)} \le a - b(\cos(\frac{\theta}{2} + \gamma) - \sin \theta)$.
\end{lemma}
\begin{proof}
    By Lemma~\ref{lem:lemma2}, we know that $0 < \cos (\frac{\theta}{2} - \gamma - \kappa) \le 1$. This implies that $1 \le \frac{1}{\cos (\frac{\theta}{2} - \gamma - \kappa)}$ and thus $-\frac{1}{\cos \left(\frac{\theta}{2} - \gamma - \kappa\right)} \le -1$. 
    
    Using that $(\cos(\frac{\theta}{2} + \gamma) - \sin \theta) > 0$ from Lemma~\ref{lem:lemma1}, we obtain:
    \begin{align*}
        -\frac{b(\cos(\frac{\theta}{2} + \gamma) - \sin \theta)}{\cos (\frac{\theta}{2} - \gamma - \kappa)} &\le -b\left(\cos\left(\frac{\theta}{2} + \gamma\right) - \sin \theta\right)
    \end{align*}
    which implies that
    \begin{align*}
        a - \frac{b(\cos(\frac{\theta}{2} + \gamma) - \sin \theta)}{\cos (\frac{\theta}{2} - \gamma - \kappa)} &\le a - b\left(\cos\left(\frac{\theta}{2} + \gamma\right) - \sin \theta\right)
    \end{align*}
    An analogous argument shows that $a-\frac{b(\cos(\frac{\theta}{2} + \gamma) - \sin \theta)}{\cos (\frac{\theta}{2} + \gamma - \kappa)} \le a - b(\cos(\frac{\theta}{2} + \gamma) - \sin \theta)$.
\end{proof}

\begin{lemma}
\label{lem:lemma4}
Let $\theta \in (0, \frac{2\pi}{7}]$ and $\gamma \in [0, \frac{\pi - 3\theta}{2})$. Then $\cos(\frac{\theta}{2} - \gamma) \ge \cos(\frac{\theta}{2} + \gamma)$.
\end{lemma}
\begin{proof}
To prove this, we show that $\cos\left(\frac{\theta}{2} - \gamma\right) -\cos\left(\frac{\theta}{2} + \gamma\right)$ is at least 0.
\begin{align*}
    \cos\left(\frac{\theta}{2} - \gamma\right) -\cos\left(\frac{\theta}{2} + \gamma\right) &= \cos\frac{\theta}{2}\cos \gamma + \sin \frac{\theta}{2} \sin \gamma - \cos \frac{\theta}{2} \cos  \gamma + \sin \frac{\theta}{2} \sin \gamma\\
    &= 2\sin \frac{\theta}{2} \sin \gamma\\
    &\ge 0 \text{ (due to the domain of $\theta$ and $\gamma$)}
\end{align*}
\end{proof}

\begin{lemma}
\label{lem:lemma5}
Let $\theta \in (0, \frac{2\pi}{7}]$, $\gamma \in [0, \frac{\pi-3\theta}{2})$, and $\kappa \in [0, \theta]$. Then $\cos(\frac{\theta}{2} - \gamma - \kappa) \ge \cos(\frac{\theta}{2} + \gamma)$.
\end{lemma}

\begin{proof}
We observe that $\cos(\frac{\theta}{2} - \kappa - \gamma) = \cos(\kappa - (\frac{\theta}{2} - \gamma))$. Note that $ -\frac{\pi}{2} < \frac{\theta}{2} - \gamma < \frac{\pi}{2}$ and that $\cos(\kappa - (\frac{\theta}{2} - \gamma))$ corresponds to the shifted $\cos \kappa$ function.
To prove the lemma, we distinguish between two cases:

\emph{Case 1:} If $\frac{\theta}{2} - \gamma \le 0$, $\cos(\kappa - (\frac{\theta}{2}-\gamma))$ corresponds to translating $\cos \kappa$ to the left by $\frac{\theta}{2} - \gamma$. Therefore, $\cos(\kappa - (\frac{\theta}{2} - \gamma))$ is decreasing over the domain of $\kappa$, which implies that $\cos(\frac{\theta}{2} - \gamma - \kappa) \ge \cos(\frac{\theta}{2} - \gamma - \theta) = \cos(\frac{\theta}{2} + \gamma)$.

\emph{Case 2:} If $\frac{\theta}{2} - \gamma > 0$, $\cos(\kappa - (\frac{\theta}{2}-\gamma))$ corresponds to translating $\cos \kappa$ to the right by $\frac{\theta}{2}-\gamma$. Therefore, for $\kappa \in (\frac{\theta}{2}-\gamma, \theta]$ the function is decreasing, so we can apply an argument analogous to that in Case 1 to prove the result in this domain.

It remains to prove the result for $\kappa \in [0, \frac{\theta}{2}-\gamma]$. In this domain, $\cos(\kappa - (\frac{\theta}{2}-\gamma))$ is increasing. Therefore to prove the result, we need to show that at $\kappa = 0$ (where $\cos(\kappa - (\frac{\theta}{2}-\gamma))$ is minimized in this domain), $\cos(\kappa - (\frac{\theta}{2}-\gamma)) > \cos(\frac{\theta}{2} + \gamma)$. After substituting $\kappa = 0$, we see that this corresponds to showing that $\cos(\frac{\theta}{2} - \gamma) \ge \cos(\frac{\theta}{2} + \gamma)$ which follows from Lemma~\ref{lem:lemma4}.
\end{proof}

\begin{lemma}
\label{lem:lemma6}
Let $\theta \in (0, \frac{2\pi}{7}]$, $\gamma \in [0, \frac{\pi-3\theta}{2})$, and $\kappa \in [0, \theta]$. Then $\cos(\frac{\theta}{2} + \gamma - \kappa) \ge \cos(\frac{\theta}{2} + \gamma)$.
\end{lemma}
\begin{proof}
We observe that $\cos(\frac{\theta}{2} - \kappa + \gamma) = \cos(\kappa - (\frac{\theta}{2} + \gamma))$. Note that $ 0 < \frac{\theta}{2} + \gamma < \frac{\pi}{2}$ and that $\cos(\kappa - (\frac{\theta}{2} + \gamma))$ corresponds to translating $\cos \kappa$ to the right by $\frac{\theta}{2}+\gamma$, since $\frac{\theta}{2}+\gamma$ is positive. 
Therefore, for $\kappa \in [0, \frac{\theta}{2}+\gamma],\text{ } \cos(\kappa - (\frac{\theta}{2} + \gamma))$ is increasing and so
$\cos(\kappa - (\frac{\theta}{2} + \gamma)) \ge \cos(-(\frac{\theta}{2} + \gamma)) = \cos(\frac{\theta}{2} + \gamma).$

For $\kappa \in (\frac{\theta}{2}+\gamma, \theta], \text{ } \cos(\kappa - (\frac{\theta}{2} + \gamma))$  is decreasing. Therefore to prove the result, we need to show that at $\kappa = \theta$ (where $\cos(\kappa - (\frac{\theta}{2}+\gamma))$ is minimized in this domain), $\cos(\kappa - (\frac{\theta}{2}+\gamma)) \ge \cos(\frac{\theta}{2} + \gamma)$. After substituting $\kappa = \theta$, we see that this corresponds to showing that $\cos(\frac{\theta}{2} - \gamma) \ge \cos(\frac{\theta}{2} + \gamma)$ which we proved in Lemma~\ref{lem:lemma4}.
\end{proof}

\subsection{Main Geometric Lemma}
Now that we have our auxiliary lemmas, we are ready to prove the main geometric lemma that we use throughout our later proofs.

\begin{lemma}
\label{lem:lemma7}
Let $k \geq 7$ be an integer, $\theta$ = $\frac{2\pi}{k}$, and $\gamma \in [0, \frac{\pi - 3\theta}{2})$. Let $p$ and $q$ be two distinct points in the plane and let $C$ be the cone of $\zeta_{k}$ such that $q \in C_{p}$. Let $r$ be a point in $C_{p}$ such that it is at least as close to $p$ as $q$ is to $p$. Then
$| rq | \le | pq |  -  (\cos (\frac{\theta}{2} + \gamma) - \sin \theta) | pr |$.
\end{lemma}
\begin{proof}
If $r = q$ then the claims holds. We assume in the rest of the proof that $r \neq q$.

Let $\ell$ be the line through $p$ and $q$. Let $s$ be the intersection of $\ell$ and the sweeping line through $r$. Let $a$ and $b$ be the intersection of the sweeping line through $r$ with the left and right side of the cone respectively. Let $x$ be the intersection of the right side of the cone and the line through $a$ perpendicular to the bisector of the cone. Finally, let $\alpha$ be the angle between the segments $pq$ and $pr$ and let $\beta$ be the angle between the segment $pr$ and either the left or right side of the cone such that $\alpha$ and $\beta$ do not overlap. We have $0 \leq \alpha, \beta \leq \theta$ and $0 \leq \alpha + \beta \leq \theta$. We distinguish two cases depending on whether $r$ lies to the left or right of $\ell$.

\emph{Case 1:} If $r$ lies to the left of $\ell$ (see Figure~\ref{fig:mainLemmaCase1}), we have that since $\triangle pax$ is isosceles, $\angle pax = \frac{\pi - \theta}{2}$. By considering $\triangle pas$, we can then deduce that $\angle asp = \frac{\pi}{2} + \frac{\theta}{2} - \gamma - (\alpha + \beta)$. Finally, by considering $\triangle prs$, we can deduce that $\angle prs = \frac{\pi}{2} - \frac{\theta}{2} + \gamma + \beta$.

\begin{figure}[h]
\includegraphics[]{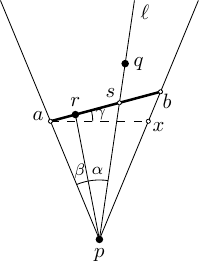}
\centering
\caption{The points and angles defined for Case 1.}
\label{fig:mainLemmaCase1}
\end{figure}

Applying the sine rule and trigonometric rewrite rules, we have:
\begin{align*}
| rs | &= | pr | \frac{\sin \alpha}{\cos (\frac{\theta}{2} - (\alpha + \beta) - \gamma)} \\
&\le | pr | \frac{\sin \theta}{\cos (\frac{\theta}{2} - (\alpha + \beta) - \gamma)} 
\end{align*}

and
\begin{align*}
| ps | &= | pr | \frac{\cos (\frac{\theta}{2} - \beta - \gamma)}{\cos (\frac{\theta}{2} - (\alpha + \beta) - \gamma)} \\
&\ge | pr | \frac{\cos (\frac{\theta}{2} + \gamma)}{\cos (\frac{\theta}{2} - (\alpha + \beta) - \gamma)} \text{(using Lemma~\ref{lem:lemma5}).}
\end{align*}

Applying the triangle inequality, we get:
\begin{align*}
    | rq | &\le | rs | + | sq | \\
    &= | rs | + | pq | - | ps | \\
    &\le | pq | + | pr | \frac{\sin \theta}{\cos (\frac{\theta}{2} - (\alpha + \beta) - \gamma)} - | pr | \frac{\cos (\frac{\theta}{2} + \gamma)}{\cos (\frac{\theta}{2} - (\alpha + \beta) - \gamma)} \\
    &= | pq | - | pr | \frac{1}{\cos \left(\frac{\theta}{2} - (\alpha + \beta) - \gamma\right)}\left(\cos \left(\frac{\theta}{2} + \gamma\right) - \sin{\theta}\right) \\
    &\le | pq | - | pr | \left(\cos \left(\frac{\theta}{2} + \gamma\right) - \sin{\theta}\right) \text{ (using Lemma~\ref{lem:lemma3}).}
\end{align*}

\emph{Case 2:} If $r$ lies to the right of $q$ (see Figure~\ref{fig:mainLemmaCase2}), we have that since $\triangle pax$ is an isosceles triangle, $\angle pxa = \frac{\pi - \theta}{2}$. This implies that $\angle axb = \frac{\pi + \theta}{2}$. By considering $\triangle abx$, we can then deduce that $\angle abx = \frac{\pi}{2} - \frac{\theta}{2} - \gamma$. We can then deduce that $\angle psb = \frac{\pi}{2} + \frac{\theta}{2} + \gamma - (\alpha + \beta)$ by considering $\triangle psb$. Finally, by considering $\triangle psr$, we can deduce that $\angle srp = \frac{\pi}{2} - \frac{\theta}{2} - \gamma + \beta$.

\begin{figure}[h]
\includegraphics[]{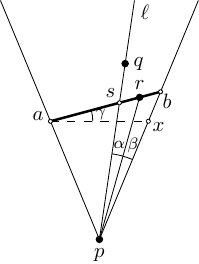}
\centering
\caption{The points and angles defined for Case 2.}
\label{fig:mainLemmaCase2}
\end{figure}

By applying the sine rule and using trigonometric rewrite rules we have:
\begin{align*}
| rs | &= | pr | \frac{\sin \alpha}{\cos (\frac{\theta}{2} - (\alpha + \beta) + \gamma)} \\
&\le | pr | \frac{\sin \theta}{\cos (\frac{\theta}{2} - (\alpha + \beta) + \gamma)} 
\end{align*}

and
\begin{align*}
| ps | &= | pr | \frac{\cos (\frac{\theta}{2} - \beta + \gamma)}{\cos (\frac{\theta}{2} - (\alpha + \beta) + \gamma)} \\
&\ge | pr | \frac{\cos (\frac{\theta}{2} + \gamma)}{\cos (\frac{\theta}{2} - (\alpha + \beta) + \gamma)} \text{(using Lemma~\ref{lem:lemma6}).}
\end{align*}

An argument identical to that of Case 1 completes the proof of this case.
\end{proof}

\section{The Unconstrained Setting}
Now that we have our tools ready, it is time to prove that the sweeping line graph is a spanner in the unconstrained setting. 

\begin{theorem}
Let k $\ge$ 7 be an integer, let $\theta = \frac{2\pi}{k}$, and let $\gamma \in [0, \frac{\pi - 3\theta}{2})$. Then the sweeping line construction produces a $t$-spanner, where t = $\frac{1}{\cos(\frac{\theta}{2} + \gamma) - \sin\theta}$.
\end{theorem}
\begin{proof}
Let $u$ and $w$ be two distinct vertices of $P$. We consider the path $u = v_{0}, v_{1}, v_{2}, ...$ that is constructed by the \emph{sweeping-routing} algorithm. We start by showing that this path terminates at $w$. Let $i \ge 0$ and assume that $v_{i} \neq w$. Hence, vertex $v_{i+1}$ exists. Consider the three vertices $v_{i}$, $v_{i+1}$, and $w$. Let $C$ be the cone such that $w \in C_{u}$. By the construction of the sweeping line graph, $v_{i+1}$ is at least as close to $v_i$ as $w$ is to $v_i$. Hence, by applying Lemma~\ref{lem:lemma7} and Lemma~\ref{lem:lemma1} we obtain:
\[| v_{i+1}w | \le | v_{i}w | - \left(\cos\left(\frac{\theta}{2} + \gamma\right) - \sin\theta\right)| v_{i}v_{i + 1} | < | v_{i}w |.\]

Hence, the vertices $v_{0}, v_{1}, v_{2}, ...$ on the path starting at $u$ are pairwise distinct, as each vertex on this path lies strictly closer to $w$ than any of its predecessors. Since the set $P$ is finite, this implies that the algorithm terminates. Therefore, the algorithm constructs a path between $u$ and $w$.

We now prove an upper bound on the length of this path. Let $m$ be the index such that $v_{m} = w$. Rearranging  $| v_{i+1}w | \le | v_{i}w | - (\cos(\frac{\theta}{2} + \gamma) - \sin\theta)| v_{i}v_{i + 1} |$, yields 
\[| v_{i}v_{i + 1} | \le \frac{1}{\cos\left(\frac{\theta}{2} + \gamma\right) - \sin\theta}(| v_{i}w | - | v_{i+1}w |),\] 
for each $i$ such that $0 \le i < m$. 

Therefore, the path between $u$ and $w$ has length
\begin{align*}
    \sum_{i=0}^{m-1} | v_{i}v_{i+1} | &\leq \frac{1}{\cos(\frac{\theta}{2} + \gamma) - \sin\theta} \sum_{i=0}^{m-1} (| v_{i}w | - | v_{i+1}w |)\\
    &= \frac{1}{\cos(\frac{\theta}{2} + \gamma) - \sin\theta}(| v_{0}w | - | v_{m}w |)\\
    &= \frac{1}{\cos(\frac{\theta}{2} + \gamma) - \sin\theta} |uw |,
\end{align*}
completing the proof. 
\end{proof}

In addition to showing that the graph is a spanner, the above proof shows that the sweeping-routing algorithm constructs a bounded length path, thus we obtain a local competitive routing algorithm. 

\begin{corollary}
Let k $\ge$ 7 be an integer, let $\theta = \frac{2\pi}{k}$, and let $\gamma \in [0, \frac{\pi - 3\theta}{2})$. Then for any pair of vertices $u$ and $w$ the sweeping-routing algorithm produces a path from $u$ to $w$ of length at most $\frac{1}{\cos(\frac{\theta}{2} + \gamma) - \sin\theta} \cdot |uw|$.
\end{corollary}

\section{The Constrained Setting}
Next, we generalize the sweeping line graph to a more general setting with the introduction of \emph{line segment constraints}. Recall that $P$ is a set of points in the plane and that $S$ is a set of line segments connecting two points in $P$ (not every point in $P$ needs to be an endpoint of a constraint in $S$, but a point in $P$ can be the endpoint of an arbitrary number of constraints). The set of constraints is planar, i.e. no two constraints intersect properly. 

Let vertex $u$ be an endpoint of a constraint $c$ and let the other endpoint be $v$. Let $C$ be the cone of $\zeta_{k}$ such that $v \in C_{u}$. The line through $c$ splits $C_u$ into two \emph{subcones} and for simplicity, we say that $v$ is contained in both of these. In general, a vertex $u$ can be an endpoint of several constraints and thus a cone can be split into several subcones (subcones are ordered in a clockwise fashion). For ease of exposition, when a cone $C_u$ is not split, we consider $C_u$ itself to be a single subcone. We use $C^{j}_u$ to denote the $j$-th subcone of $C_u$.

Recall that for any vertex $x$ in a cone, we defined $x_{\gamma}$ to be the intersection of the left-side of the cone and the sweeping line through $x$. We define the constrained sweeping line graph (see Figure~\ref{fig:defnConstrained}):
\begin{definition}[Constrained sweeping line graph] Given a set of points $P$ in the plane, a plane set $S$ of line segment constraints connecting pairs of points in $P$, an integer $k \ge 7$, $\theta = \frac{2\pi}{k}$, and $\gamma \in [0, \frac{\pi - 3\theta}{2})$. The constrained sweeping line graph is defined as follows:
\begin{enumerate}
    \item The vertices of the graph are the points of $P$.
    \item For each vertex $u$ of $P$ and for each \emph{subcone} $C^j_u$ that contains one or more vertices of $P \setminus \{u\}$ visible to $u$, the spanner contains the undirected edge $(u, r)$, where $r$ is the vertex in $C^j_u \cap P \setminus \{u\}$, which is visible to $u$ and minimizes $|ur_{\gamma}|$. This vertex $r$ is referred to as the \emph{closest} visible vertex in this subcone of $u$.
\end{enumerate}
\end{definition}

\begin{figure}[h]
\includegraphics{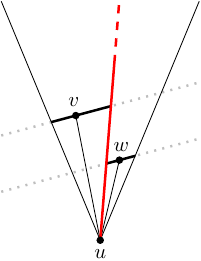}
\centering
\caption{The edges in a cone of the constrained sweeping line graph. The thick red segment represents a constraint. The sweeping line of a subcone is a thick black segment inside the subcone and grey dotted outside, as vertices outside the subcone as ignored.}
\label{fig:defnConstrained}
\end{figure}

To prove that the above graph is a spanner, we need three lemmas. A proof of Lemma~\ref{lem:convexChainConstrained} can be found in~\cite{BFRV12} (see also Figure~\ref{fig:visible}). 

\begin{lemma}[\cite{BFRV12}]
\label{lem:convexChainConstrained}
Let $u$, $v$, and $w$ be three arbitrary points in the plane such that $uw$ and $vw$ are visibility edges and $w$ is not the endpoint of a constraint intersecting the interior of triangle $uvw$. Then there exists a convex chain of visibility edges (different from the chain consisting of $uw$ and $wv$) from $u$ to $v$ in triangle $uvw$, such that the polygon defined by $uw$, $wv$ and the convex chain is empty and does not contain any constraints.
\end{lemma}

\begin{figure}[h]
\includegraphics{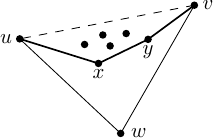}
\centering
\caption{Illustration of Lemma~\ref{lem:convexChainConstrained}.}
\label{fig:visible}
\end{figure}

\begin{lemma}
\label{lem:visibilityEdgeConstrained}
Let $u$ and $w$ be two distinct vertices in the constrained sweeping line graph such that $uw$ is a visibility edge and let $C$ be the cone of $\zeta_{k}$ such that $w \in C_{u}$. Let $v_{0}$ be the closest visible vertex in the subcone of $C_{u}$ that contains $w$. Let $\ell$ be the line through $u$ and $w$. Let $s$ be the intersection of $\ell$ and the sweeping line through $v_{0}$. Then $v_{0}s$ is a visibility edge.
\end{lemma}
\begin{proof}
We use a proof by contradiction (see Figure~\ref{fig:constrained} for an example layout). Assume $v_{0}s$ is not a visibility edge. Then there must be a line segment constraint intersecting $v_{0}s$. This implies that one of its endpoints lies in $\triangle uv_{0}s$, as $uw$ and $uv_{0}$ are visibility edges and thus the constraint cannot pass through them. Applying Lemma~\ref{lem:convexChainConstrained} on $\triangle uv_{0}s$, which contains at least one vertex in its interior, implies that there exists a vertex that is visible to $u$ and closer to $u$ than $v_{0}$ (in particular the first vertex hit by the sweeping line starting from $u$), contradicting that $v_0$ is the closest visible vertex. Therefore, no constraint intersects $v_{0}s$.
\end{proof}

\begin{figure}[h]
\includegraphics{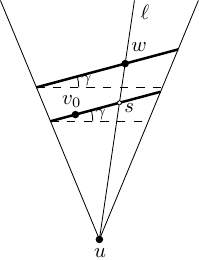}
\centering
\caption{An example layout of $p$, $q$, $v_0$, and $s$.}
\label{fig:constrained}
\end{figure}

The following lemma ensures that we can apply induction later. 

\begin{lemma}
\label{lem:segmentLengths}
Let $k \ge 7$ be an integer, let $\theta = \frac{2\pi}{k}$, and let $\gamma \in [0, \frac{\pi - 3\theta}{2})$. Let $u$ and $w$ be two distinct vertices in the constrained sweeping line graph and let $C$ be the cone of $\zeta_{k}$ such that $w \in C_{u}$. Let $v_{0}$ be a vertex in $C_{u}$ such that it is the closest visible vertex to $u$. Let $\ell$ be the line through $uw$ and let $s$ be the intersection of $\ell$ and the sweeping line through $v_{0}$. Then $| v_{0}s | < | uw |$, $| sw | < | uw |$, and $| v_{0}w | < | uw |$.
\end{lemma}
\begin{proof}
Refer to Figure~\ref{fig:constrained} for an example layout.

We first show that $| v_{0}s | < | uw |$:
By applying Lemma~\ref{lem:lemma7} to $u$, $s$, and $v_0$, we obtain that 
$| v_{0}s | \le | us | - (\cos(\frac{\theta}{2} + \gamma) - \sin\theta) \cdot| uv_{0} |$. Using Lemma~\ref{lem:lemma1}, this implies that $| v_{0}s | < | us |$. Finally, using the fact that $us$ is contained in $uw$, we conclude that $| v_{0}s | < | uw |$.

Next, since $sw$ is contained in $uw$, it follows that $| sw | < | uw |$.

Finally, we argue that $| v_{0}w | < | uw |$:
By applying Lemma~\ref{lem:lemma7} to $u$, $w$, and $v_0$, we obtain that $| v_{0}w | \le | uw | - (\cos(\frac{\theta}{2} + \gamma) - \sin\theta) \cdot | uv_{0} | $. We can then apply Lemma~\ref{lem:lemma1} to obtain that
$| v_{0}w | < | uw |$.
\end{proof}

We are now ready to prove that the constrained sweeping line graph is a spanner of the visibility graph. 

\begin{theorem}
\label{theo:constrained}
Let $k \geq 7$ be an integer, let $\theta = \frac{2\pi}{k}$, and let $\gamma \in [0, \frac{\pi - 3\theta}{2})$. Let $u$ and $w$ be two distinct vertices in the plane that can see each other. There exists a path connecting $u$ and $w$ in the constrained sweeping line graph of length at most $\frac{1}{\cos(\frac{\theta}{2} + \gamma) - \sin(\theta)} \cdot | uw |$.
\end{theorem}
\begin{proof}
Let $C$ be the cone of $\zeta_{k}$ such that $w \in$ $C_{u}$.
We prove the theorem by induction on the rank of the pairs of vertices that can see each other, based on the Euclidean distance between them.
Our inductive hypothesis is the following:
$\delta(u,w) \le \frac{1}{\cos(\frac{\theta}{2} + \gamma) - \sin\theta} | uw |$, where $u$ and $w$ are two distinct vertices that can see each other and $\delta(u,w)$ is the length of the shortest path between them in the constrained sweeping line graph.

\emph{Base case:}
In this case, $u$ and $w$ are the Euclidean closest visible pair. If there exists an edge between $u$ and $w$, then $\delta(u, w) = | uw | \le \frac{1}{\cos(\frac{\theta}{2} + \gamma) - \sin\theta} | uw |$, so the induction hypothesis holds. 

It remains to show that, indeed, there exists an edge between the Euclidean closest visible pair. We prove this using contradiction. Assume that there is no edge between $u$ and $w$. Then there must exist a vertex $v_{0}$ in the subcone $C^j_u$ that contains $w$ that has an edge to $u$ in the constrained sweeping line graph. Let $\ell$ be the line through $uw$. Let $s$ be the intersection of $\ell$ and the sweeping line through $v_{0}$. Note that $sw$ is a visibility edge, due to $uw$ being a visibility edge, and $v_{0}s$ is a visibility edge, by Lemma~\ref{lem:visibilityEdgeConstrained}. Therefore, by applying Lemma~\ref{lem:convexChainConstrained}, there exists a convex chain $v_{0}, v_{1}, ..., v_{m-1}, v_{m} = w$ of visibility edges from $v_{0}$ to $w$ inside $\triangle v_0 s w$.

By applying Lemma~\ref{lem:segmentLengths} using $u$, $w$, and $v_{0}$, we infer that all sides of $\triangle v_0 s w$ have length less than the Euclidean distance between $u$ and $w$. Since the convex chain is contained in this triangle, it follows that any pair of consecutive vertices along it has a smaller Euclidean distance than the Euclidean distance between $u$ and $w$. This contradicts that $uw$ is the closest Euclidean pair of visible vertices.

\emph{Induction step:}
We assume that the induction hypothesis holds for all pairs of vertices that can see each other and whose Euclidean distance is smaller than the Euclidean distance between $u$ and $w$.

If $uw$ is an edge in the constrained sweeping line graph, the induction hypothesis follows by the same argument as in the base case. If there is no edge between $u$ and $w$, let $v_{0}$ be the closest visible vertex to $u$ (using the sweeping line) in the subcone $C^j_u$ that contains $w$. By construction, $(u,v_0)$ is an edge of the graph. Let $\ell$ be the line passing through $u$ and $w$. Let $s$ be the intersection of $\ell$ and the sweeping line through $v_{0}$ (see Figure~\ref{fig:constrained}). By definition, $\delta(u,w) \le | uv_{0} | + \delta(v_{0}, w)$.

We know that $sw$ is a visibility edge, since $uw$ is a visibility edge, and we know $v_{0}s$ is a visibility edge by Lemma~\ref{lem:visibilityEdgeConstrained}. Therefore, by Lemma~\ref{lem:convexChainConstrained} there exists a convex chain $v_{0},...,v_{m} = w$ of visibility edges inside $\triangle v_0 s w$ connecting $v_{0}$ and $w$. Applying Lemma~\ref{lem:segmentLengths} to the points $u$, $v_{0}$, and $w$, we infer that each side of $\triangle v_{0}sw$ has length smaller than $|uw|$. Therefore, we can apply induction to every visibility edge along the convex chain from $v_0$ to $w$, as each has length smaller than $|uw|$. Therefore, 
\begin{align*}
\delta(u,w) 
&\le | uv_{0} | + \sum_{i = 0}^{m-1} \delta(v_{i},v_{i + 1})  \\
&\le | uv_{0} | + \frac{1}{\cos(\frac{\theta}{2} + \gamma) - \sin\theta}\sum_{i = 0}^{m-1} | v_{i}v_{i + 1} | &\text{ (by induction hypothesis)}\\
&\le | uv_{0} | + \frac{1}{\cos(\frac{\theta}{2} + \gamma) - \sin\theta} (| v_{0}s | + | sw |) &\text{ (since the chain is convex)} 
\end{align*}

Finally, we apply Lemma~\ref{lem:lemma7}, using $r = v_{0}$, $q = s$, and $p = u$. This gives us that $| v_{0}s | \le | us | - | uv_{0} | \left(\cos\left(\frac{\theta}{2} + \gamma\right) - \sin \theta\right)$, which can be rewritten to $| uv_{0} | + | v_{0}s |/(\cos(\frac{\theta}{2} + \gamma) - \sin \theta) \le | us |/(\cos(\frac{\theta}{2} + \gamma) - \sin \theta)$.
By adding $| sw |/(\cos(\frac{\theta}{2} + \gamma) - \sin\theta)$ to both sides and the fact that $|us| + |sw| = |uw|$, we obtain:
\begin{align*}
| uv_{0} | + \frac{1}{\cos(\frac{\theta}{2} + \gamma) - \sin \theta}(| v_{0}s | + | sw |)
&\le \frac{1}{\cos(\frac{\theta}{2} + \gamma) - \sin \theta} (| us | + | sw |)\\
&= \frac{1}{\cos(\frac{\theta}{2} + \gamma) - \sin \theta} |uw|.
\end{align*}

Hence, we conclude that $\delta(u,w) \le \frac{1}{\cos(\frac{\theta}{2} + \gamma) - \sin \theta} | uw |$.
\end{proof}

We note that the above theorem implies that the constrained sweeping line graph also provides a path with bounded spanning ratio with respect to the shortest path in Vis$(P,S)$ for \emph{every} pair of points, since every pair of consecutive vertices along this shortest path is by definition visible. 

\section{Polygonal Obstacles}
Finally, we generalize the result from the previous section to more complex obstacles. Recall that in this setting $S$ is a finite set of polygonal obstacles, where each corner of each obstacle is a point in $P$, such that no two obstacles intersect, and that each point is part of at most one polygonal obstacle and occurs at most once along its boundary. No obstacle has any points of $P$ in its interior. 

In a nutshell, the defintions used in this section are highly similar to those of the previous section: As in the constrained setting, the line segment between two visible vertices is called a \emph{visibility edge} and the \emph{visibility graph} of a point set $P$ and a set of polygonal obstacles $S$ is the complete graph on $P$ excluding all the edges that properly intersect some obstacle. Cones that are split are considered to be subcones of the original cone. Note that since $S$ consists of vertex-disjoint simple polygons, a cone can be split into at most two subcones. By focusing on the subcones, the polygonal-constrained sweeping line graph is defined analogously to the constrained sweeping line graph: for each subcone $C^{j}_u$ of each vertex $u$, we add an undirected edge between $u$ and the closest vertex in that subcone that can see $u$, where the distance is measured along the left side of the original cone of $u$ (not the left side of the subcone).

For completeness, we also include the full formal definition. Let vertex $u$ be a corner of an obstacle $o$ and let $y$ and $z$ be the corners preceding and following $u$ on the cyclic order of the corners of $o$. Let $C$ be the cone of $\zeta_{k}$ such that $y \in C_{u}$ and let $C'$ be the cone of $\zeta_{k}$ such that $z \in C'_{u}$. If $C=C'$, obstacle $o$ splits $C_u$ into two \emph{subcones}, one clockwise from the obstacle and one counterclockwise from the obstacle (see Figure~\ref{fig:polygonalConstrainedCones}a). Since each vertex $u$ can be a corner of at most one obstacle, only one of its cones can be split this way. We note that if $C \neq C'$, then this simply means that the visible region in those two cones is narrower than those of the other cones, but they are otherwise treated the same way as the other cones (see Figure~\ref{fig:polygonalConstrainedCones}b). For ease of exposition, when a cone $C_u$ is not split, we consider $C_u$ itself to be a single subcone. We use $C^{j}_u$ to denote the $j$-th subcone of $C_u$.

\begin{figure}[h]
\includegraphics{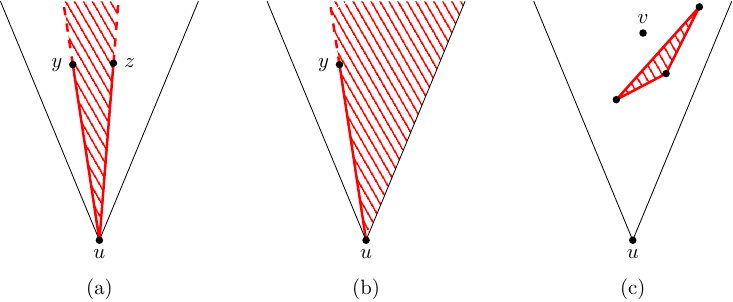}
\centering
\caption{The different ways in which a polygonal obstacle can affect the construction: (a) a cone is split by the obstacle, (b) the visible region in a cone is made narrower by an obstacle, (c) vertex $u$ and $v$ cannot see each other because of the obstacle, but the obstacle does not create subcones for $u$, as $u$ is not one of its corners.}
\label{fig:polygonalConstrainedCones}
\end{figure}

Recall that for any vertex $x$ in a cone, we defined $x_{\gamma}$ to be the intersection of the left-side of the cone and the sweeping line through $x$. We define the polygonal-constrained sweeping line graph (see Figure~\ref{fig:defnPolygonalConstrained}):
\begin{definition}[Polygonal-constrained sweeping line graph] Given a set of points $P$ in the plane, a non-intersecting set $S$ of simple polygonal obstacles whose corners are points in $P$, an integer $k \ge 7$, $\theta = \frac{2\pi}{k}$, and $\gamma \in [0, \frac{\pi - 3\theta}{2})$. The polygonal-constrained sweeping line graph is defined as follows:
\begin{enumerate}
    \item The vertices of the graph are the points of $P$.
    \item For each vertex $u$ of $P$ and for each \emph{subcone} $C^j_u$ that contains one or more vertices of $P \setminus \{u\}$ visible to $u$, the spanner contains the undirected edge $(u, r)$, where $r$ is the vertex in $C^j_u \cap P \setminus \{u\}$, which is visible to $u$ and minimizes $|ur_{\gamma}|$. This vertex $r$ is referred to as the \emph{closest} visible vertex in this subcone of $u$.
\end{enumerate}
\end{definition}

\begin{figure}[h]
\includegraphics{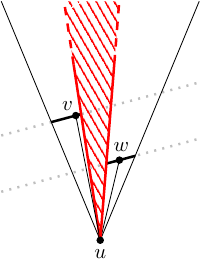}
\centering
\caption{The edges in a cone of the polygonal-constrained sweeping line graph. The red area represents (part of) an obstacle. The sweeping line of a subcone is a thick black segment inside the subcone and grey dotted outside, as vertices outside the subcone as ignored.}
\label{fig:defnPolygonalConstrained}
\end{figure}

We now introduce modifications to Lemmas~\ref{lem:convexChainConstrained} and ~\ref{lem:visibilityEdgeConstrained} to make them suited for polygonal obstacles. A proof of Lemma~\ref{lem:convexChainPolygonal} can be found in \cite{polygonallemma}.
\begin{lemma}[\cite{polygonallemma}]
\label{lem:convexChainPolygonal}
Let $u$, $v$, and $w$ be three points where $(w,u)$ and $(u,v)$ are both visibility edges and $u$ is not a vertex of any polygonal obstacle $P$ where the open polygon $P'$ of $P$ intersects $\triangle wuv$. The area $A$, bounded by $(w,u)$, $(u,v)$, and a convex chain formed by visibility edges between $w$ and $v$ inside $\triangle wuv$, does not contain any vertices and is not intersected by any obstacles.
\end{lemma}

\begin{lemma}
\label{lem:visibilityEdgePolygonal}
Let $u$ and $w$ be two distinct vertices in the polygonal-constrained sweeping line graph such that $uw$ is a visibility edge and let $C$ be the cone of $\zeta_{k}$ such that $w \in C_{u}$. Let $v_{0}$ be the closest visible vertex in the subcone of $C_{u}$ that contains $w$. Let $\ell$ be the line through $u$ and $w$. Let $s$ be the intersection of $\ell$ and the sweeping line through $v_{0}$. Then $v_{0}s$ is a visibility edge.
\end{lemma}
\begin{proof}
This proof is analogous to the proof of Lemma~\ref{lem:visibilityEdgeConstrained}.
\end{proof}

Using these two modified lemmas, we can prove that the polygonal-constrained sweeping line graph is a spanner of the visibility graph. 

\begin{theorem}
Let $k \geq 7$ be an integer, let $\theta = \frac{2\pi}{k}$, and let $\gamma \in [0, \frac{\pi - 3\theta}{2})$. Let $u$ and $w$ be two distinct vertices in the plane that can see each other. There exists a path connecting $u$ and $w$ in the polygonal-constrained sweeping line graph of length at most
$\frac{1}{\cos(\frac{\theta}{2} + \gamma) - \sin(\theta)} \cdot | uw |$.
\end{theorem}
\begin{proof}
The proof is analogous to the proof of Theorem~\ref{theo:constrained}. The only changes required are that the uses of Lemma~\ref{lem:convexChainConstrained} and Lemma~\ref{lem:visibilityEdgeConstrained} are replaced with Lemma~\ref{lem:convexChainPolygonal} and Lemma~\ref{lem:visibilityEdgePolygonal} respectively. Note that all other arguments still hold, as they are arguments based on Euclidean distance, rather than the specific shape of the (straight-line) obstacles.
\end{proof}

We note that, like the constrained sweeping line graph in the previous section, the above theorem implies that the polygonal-constrained sweeping line graph also provides a path with bounded spanning ratio with respect to the shortest path in Vis$(P,S)$ for \emph{every} pair of points.

\section{Conclusion}
We showed that the sweeping line construction produces a spanner in the unconstrained, constrained, and polygonal obstacle settings.
These graphs are a generalization of $\Theta$-graphs and thus we also showed that every $\Theta$-graph with at least 7 cones is a spanner in the presence of polygonal obstacles.
We also note that the proof in the unconstrained case immediately implied a local routing algorithm with competitive ratio equal to (the current upper bound of) the spanning ratio. 

Our proofs rely on Lemma~\ref{lem:lemma7}, which bounds the length of the inductive part of our path. We conjecture that any proof strategy that uses induction must satisfy a condition similar to Lemma~\ref{lem:lemma7} in order to upper bound the spanning ratio. An analogous argument could then be applied to prove the construction to be a spanner in all three settings using the methods from this paper (i.e., finding a vertex $v_0$ satisfying the conditions of Lemmas~\ref{lem:visibilityEdgeConstrained} and~\ref{lem:segmentLengths}). This would greatly simplify spanner construction for the constrained and polygonal obstacle settings, by putting the focus on the simpler unconstrained setting. In particular, we conjecture that the strategy described in this paper can be applied to generalize the known results for Yao-graphs. 

Other open problems include improving the spanning ratios presented in this paper or showing that this is not possible by providing matching lower bounds. Another direction for future work is extending the results to include sweeping line graphs with 4, 5, or 6 cones. Since inductive arguments do not (easily) apply when there are fewer than 7 cones, giving a general approach as done in this paper will be challenging. 

\bibliography{references}

\end{document}